\documentclass[11pt,a4paper]{amsart}
\usepackage[top=4cm,bottom=4cm,left=3.5cm,right=3.5cm]{geometry}

\usepackage{subcaption}
\usepackage{tabularx}
\usepackage{amsmath,amssymb,amsthm}
\usepackage[english]{babel}
\usepackage{mathtools}
\usepackage{enumerate}
\usepackage{algorithm}
\usepackage[noend]{algpseudocode}
\usepackage[numbers]{natbib}
\usepackage[hidelinks]{hyperref}
\usepackage{tikz}
\usepackage{tkz-graph}
\usepackage{tikz-cd}
\tikzset{%
    symbol/.style={%
        ,draw=none
        ,every to/.append style={%
            edge node={node [sloped, allow upside down, auto=false]{$#1$}}}
    }
}

\usepackage{xcolor}
\usepackage{todonotes}

\usepackage{algorithm}
\usepackage[noend]{algpseudocode}

\usepackage{algorithm}
\usepackage[noend]{algpseudocode}

\newtheorem{theorem}{Theorem}

\theoremstyle{definition}

\theoremstyle{remark}
\newtheorem{remark}[theorem]{Remark}

\numberwithin{equation}{section}

\newcommand\tab[1][1cm]{\hspace*{#1}}

\begin{document}

\title[Generalization of the Ball-Collision Algorithm]{Generalization
  of the Ball-Collision Algorithm}

\author[C. Interlando]{Carmelo Interlando}
\address{Department of Mathematics and Statistics \\
  San Diego State University \\
  San Diego, CA   92182-7720\\
} \email{carmelo.interlando@sdsu.edu}

\author[K. Khathuria]{Karan Khathuria}
\address{Institute of Mathematics\\
  University of Zurich\\
  Winterthurerstrasse 190\\
  8057 Zurich, Switzerland\\
} \email{karan.khathuria@math.uzh.ch}

\author[N. Rohrer]{Nicole Rohrer}
\address{Institute of Mathematics\\
  University of Zurich\\
  Winterthurerstrasse 190\\
  8057 Zurich, Switzerland\\
} \email{nicole.rohrer@uzh.ch}

\author[J. Rosenthal]{Joachim Rosenthal}
\address{Institute of Mathematics\\
  University of Zurich\\
  Winterthurerstrasse 190\\
  8057 Zurich, Switzerland\\
} \email{rosenthal@math.uzh.ch}

\author[V. Weger]{Violetta Weger}
\address{Institute of Mathematics\\
  University of Zurich\\
  Winterthurerstrasse 190\\
  8057 Zurich, Switzerland\\
} \email{violetta.weger@math.uzh.ch}

\thanks{The fourth author is thankful to Swiss National Science
  Foundation grant number 169510.}
\subjclass[2010]{}

\keywords{Coding Theory; ISD; Ball-Collision.}

\begin{abstract}
  In this paper we generalize the Ball-Collision Algorithm by
  Bernstein, Lange, Peters from the binary field to a general finite
  field. We also provide a complexity analysis and compare the
  asymptotic complexity to other generalized information set decoding
  algorithms.\end{abstract}

\maketitle

\section{Introduction}
\label{sec:introduction}

Since 1978 it has been known that decoding a random linear code is an
NP-complete problem, this was shown in \cite{berlekamp} by Berlekamp,
McEliece and van Tilborg.  Therefore the interesting task arises of
finding the complexity of decoding a random linear code using the best
algorithms available. Until today two main methods for decoding have
been proposed: information set decoding (ISD) and the generalized
birthday algorithm (GBA). The ISD is more efficient if the decoding
problem has only a small number of solutions, whereas GBA is efficient
when there are many solutions. Also other ideas such as statistical
decoding \cite{aljabri}, gradient decoding \cite{barg} and supercode
decoding \cite{barg2} have been proposed but fail to outperform ISD
algorithms.  An ISD algorithm is given a corrupted codeword and
recovers the message or equivalently finds the error vector.  ISD
algorithms are often formulated via the parity check matrix, since it
is enough to find a vector of a certain weight which has the same
syndrome as the corrupted codeword, this problem is also referred to
as the syndrome decoding problem.  ISD algorithms are based on a
decoding algorithm proposed by Prange \cite{prange} in 1962 and their
structures do not change much from the original: as a first step one
chooses an information set, then Gaussian elimination brings the
parity check matrix in a standard form and assuming that the errors
are outside of the information set, these row operations on the
syndrome will exploit the error vector, if the weight does not exceed
the given error capacity.

The problem of decoding random linear codes has recently been
receiving prominence with the proposal of using code-based public key
cryptosystems for an upcoming post-quantum cryptographic public key
standard.  The idea of using linear codes in public key
cryptography was first formulated by Robert
McEliece~\cite{mc78}. Since the publication of McEliece a large amount of
research has been done and the interested reader will find more
information in a recent survey~\cite{bo16}.

If the secret code is hidden well enough an adversary who wants to
break a code-based cryptosystem encounters the decoding problem of a
random linear code. It is therefore of crucial importance to
understand the complexity of the best algorithms capable of decoding a
general linear code.

The ISD algorithms were often considered when proposing a variant of
the McEliece cryptosystem, to find the key size needed for a given
security level.  ISD algorithms hence do not break a code-based
cryptosystem but they determine the choice of secure parameters.
Since some of the new proposals (for example \cite{ba16,dags,kh18})
involve codes over general finite fields, having efficient ISD
algorithms generalized to $\mathbb{F}_q$ is an essential problem.

Bernstein, Lange and Peters found a clever improvement of the ISD
algorithm which they called ball-collision decoding~\cite{ballcoll}.
The algorithm of Bernstein et. al. was presented for random binary
linear codes. The main contribution of our paper is a generalization
of the ball-collision decoding algorithm to arbitrary finite fields.

The paper is structured as follows: in Section \ref{previous_work} we
discuss the previous work on ISD algorithms focusing on those which
have been generalized to an arbitrary finite field. In Section
\ref{sec2} we describe the ball-collision algorithm over the binary
field and the notations and concepts involved in the algorithm. In
Section \ref{sec3} we present the ball-collision algorithm over
$\mathbb{F}_q$ and in Section \ref{sec:complexity} we perform the
complexity analysis of our algorithm including numerical parameter
optimization and asymptotic analysis.

\section{Related work} \label{previous_work} Many improvements have
been suggested to Prange's simplest form of ISD (see for example
\cite{chabanne, canteautsendrier, chabaud, dumer, kruk, leon,
  vantilborg}), they can be split into two types: improvements on the
Gaussian elimination step and a more probable and elaborated weight
distribution of the error vector.  The prior includes the work of
Canteaut and Chabaud \cite{canteaut}, where they show that the
information set should not be taken at random after one unsuccessful
iteration, but rather a part of the previous information set should be
reused and therefore a part of the Gaussian elimination step is
already performed. Whereas Finiasz and Sendrier \cite{sendrier} showed
that a complete Gaussian elimination is not necessary, both of these
improvements help to bring the cost of the Gaussian elimination step
down.


Now we focus on the second type of improvements, which were first
proposed for codes over the binary field and then later generalized
over an arbitrary finite field. The first improvement of Prange's ISD
was by Lee-Brickell \cite{leebrickell} in 1988, where in the
information set $p$ errors are assumed and $t-p$ outside. In 1993
Stern \cite{stern} proposed to partition the information set in to two
sets and ask for $p$ errors in each part and $t-2p$ errors outside the
information set. The generalization of both Lee-Brickell and Stern's
algorithm to a general finite field $\mathbb{F}_q$ were performed by
Peters \cite{peters} in 2010.

Niebuhr, Persichetti, Cayrel, Bulygin and Buchmann \cite{niebuhr} in
2010 improved the performance of ISD algorithms over $\mathbb{F}_q$
based on the idea of Finiasz-Sendrier \cite{sendrier} to allow the
errors to overlap in the information set.

In the past 10 years many other improvements were proposed for ISD
over $\mathbb{F}_2$. Namely, the ball-collision algorithm by
Bernstein, Lange and Peters \cite{ballcoll} in 2011, which splits the
information set in two sets, having $p_1$ and $p_2$ errors in them and
also splits the rest of the positions into three disjoint sets, having
$q_1, q_2$ and $t-p_1-p_2-q_1-q_2$ errors respectively. The
algorithm's name comes from a collision check, which builds the most
crucial part of the algorithm.

Later in 2011 May, Meurer and Thomae \cite{may} proposed an
improvement using the representation technique introduced by
Howgrave-Graham and Joux \cite{howgrave}. To this algorithm Becker,
Joux, May and Meurer \cite{becker} (BJMM) in 2012 introduced further
improvements. In the same year Meurer in his dissertation
\cite{meurer} proposed a new generalized ISD algorithm based on these
two papers.

In 2015, May-Ozerov \cite{mo} used the nearest neighbor algorithm to
improve the BJMM version of ISD. In 2016, Hirose \cite{hirose}
generalized the nearest neighbor algorithm over $\mathbb{F}_q$ and
applied it to the generalized Stern algorithm. Later in 2017, this was
applied to generalized BJMM algorithm by Gueye, Klamti and Hirose
\cite{klamti}.

In this paper we provide the missing generalization of the
ball-collision algorithm. The order of the complexities of ISD
algorithms over $\mathbb{F}_2$ is consistent also with their
generalizations over $\mathbb{F}_q$.

\section{Preliminaries}\label{sec2}

\subsection{Notation}

We first want to fix some notation: let $q$ be a prime power and let
$n, k, t \in\mathbb{N}$ be positive integers, such that $k,t<n$. We
will denote by $\mathbf{1}_n$ the $n \times n$ identity matrix.\\

For an $m\times n$ matrix $A$ and a set $S\subseteq\{1, ..., n\}$ of
size $k$, we denote by $A_S$ the $m\times k$ matrix consisting of the
columns of $A$ indexed by $S$.\\

For a set $S\subseteq \{1, ..., n\}$ of size $k$, we denote by
$\mathbb{F}_q^n(S)$ the vectors in $\mathbb{F}_q^n$ having support in
$S$. The projection of $x\in\mathbb{F}_q^n(S)$ to $\mathbb{F}_q^k$ is
then canonical and denoted by $\pi_S(x)$.\\

On the other hand we denote by $\sigma_S$ the canonical embedding of a
vector $x\in\mathbb{F}_q^k$ into $\mathbb{F}_q^n(S)$, where
$S\subseteq\{1, ..., n\}$ is again of size $k$.\\
	
For an $[n,k]$ linear code $\mathcal{C}$ over $\mathbb{F}_q$ we denote
by $H$ be the parity check matrix of size $(n-k)\times n$ and by $G$
the $k \times n$ generator matrix. We denote the Hamming weight of a
vector $x \in \mathbb{F}_q^n$, by $w(x)$.  The corrupted codeword
$c \in \mathbb{F}_q^n$ is given by $c =mG+e$, where
$m \in \mathbb{F}_q^k$ is the message and $e \in \mathbb{F}_q^n$ is
the error vector. The syndrome of $c$ is then defined as $s=Hc^\top$
and coincides with the syndrome of the error vector, since
$Hc^\top = H(mG+e)^\top = HG^\top m^\top + He^\top =He^\top$.

\subsection{Ball-collision algorithm over the binary field}

In what follows we describe the ball-collision algorithm over the
binary proposed in \cite{ballcoll} by Bernstein, Lange and Peters.


\begin{algorithm}[h!]
  \caption{Ball-collision over the binary}\label{ballcollalgo}
  \begin{flushleft}
    Input: The $(n-k)\times n$ parity check matrix $H$, the syndrome
    $s\in\mathbb{F}_2^{n-k}$ and the positive integers
    $p_1,\, p_2,\, q_1,\, q_2,\, k_1,\, k_2,\, \ell_1,\, \ell_2 \in
    \mathbb{Z}$,
    such that $k=k_1+k_2$, $p_i\leq k_i$, $q_i\leq \ell_i$ and
    $t-p_1-p_2-q_1-q_2\leq n-k-\ell_1-\ell_2$.\\
    Output: $e\in\mathbb{F}_2^n$ with $He^\top=s$ and $w(e)=t$.
  \end{flushleft}
  \begin{algorithmic}[1]
    \State Choose $I\subseteq\{1, ...,n\}$ a uniform random
    information set of size $k$.  \State Choose a uniform random
    partition of $I$ into disjoint sets $X_1$ and $X_2$ of size $k_1$
    and $k_2=k-k_1$ respectively.  \State Choose uniform random
    partition of $Y=\{1, ..., n\}\setminus I$ into disjoint sets
    $Y_1, Y_2$ and $Y_3$ of sizes $\ell_1,\ell_2$ and
    $\ell_3=n-k-\ell_1-\ell_2$.  \State Find an invertible matrix
    $U\in\mathbb{F}_2^{(n-k)\times(n-k)}$ such that
    $(UH)_Y=\mathbf{1}_{n-k}$ and $(UH)_I=\begin{bmatrix} A_1 \\ A_2
    \end{bmatrix}$,
    where $A_1\in\mathbb{F}_2^{(\ell_1+\ell_2)\times k}$ and
    $A_2\in\mathbb{F}_2^{\ell_3\times k}$.  \State Compute
    $Us=\begin{bmatrix} s_1 \\s_2
    \end{bmatrix}$
    with $s_1\in\mathbb{F}_2^{\ell_1+\ell_2}$ and
    $s_2\in\mathbb{F}_2^{\ell_3}$.  \State Compute the set $S$
    consisting of all triples
    $(A_1(\pi_I(x_1))+\pi_{Y_1 \cup Y_2}(y_1), x_1, y_1)$, where
    $x_1\in\mathbb{F}_2^n(X_1)$, $w(x_1)=p_1$ and
    $y_1\in\mathbb{F}_2^n(Y_1)$, $w(y_1)=q_1$.  \State Compute the set
    $T$ consisting of all triples
    $(A_1\pi_I(x_2)+\pi_{Y_1 \cup Y_2}(y_2)+s_1, x_2, y_2)$, where
    $x_2\in\mathbb{F}_2^n(X_2)$, $w(x_2)=p_2$ and
    $y_2\in\mathbb{F}_2^n(Y_2)$, $w(y_2)=q_2$ .  \For{each
      $(v, x_1, y_1)\in S$} \For{each $(v, x_2, y_2) \in T$}
    \If{$w(A_2(\pi_I(x_1+x_2))+s_2)=t-p_1-p_2-q_1-q_2$: } \Statex
    \tab[1.5cm] Output:
    $e=x_1+y_1+x_2+y_2+\sigma_{Y_3}(A_2(\pi_I(x_1+x_2))+s_2)$.  \Else
    \ Start over with Step 1 and a new selection of $I$.
    \EndIf
    \EndFor
    \EndFor
  \end{algorithmic}
\end{algorithm}

\begin{remark}
  Note that if $H$ is already in standard form, then $I=\{1, ..., k\}$
  and $U=\mathbf{1}_{n-k}$. In this case $H$ and $s$ can be written as
  \begin{equation*}
    (UH=)H=\begin{bmatrix}
      A_1 & \mathbf{1}_{\ell_1+\ell_2} &0 \\
      A_2 & 0 &\mathbf{1}_{\ell_3}
    \end{bmatrix}, \ (Us=)s=\begin{bmatrix} s_1 \\s_2
    \end{bmatrix}.
  \end{equation*}
\end{remark}


\subsection{Concepts}

There are a few concepts for computing the complexity of the
ball-collision algorithm introduced in \cite{ballcoll} that we will
use and present beforehand. In general the complexity of an ISD attack
consists of the cost of one iteration times the expected number of
iterations. The cost in the following refers to operations,
i.e. additions or multiplications, over the given field.

The success probability over the binary is usually given by having
chosen the correct weight distribution of the error vector. For
example let the error vector be of length $n$ having weight $t$, now
we assume that the error vector has weight $p$ in the information set,
i.e. in $k$ bits and the rest is redundant, then the success
probability is given by
\begin{equation*}
  \binom{k}{p} \binom{n}{t}^{-1}.
\end{equation*}
This will not change over $\mathbb{F}_q$, since the algorithm runs
through all elements in the finite field having support in those
chosen sets.

The concept of intermediate sums is important whenever one wants to
compute something for all vectors in a certain space. For example we
are given a $k \times n$ matrix $A$ and want to compute $Ax^\top$ for
all $x \in \mathbb{F}_2^n$, of weight $t$. This would usually cost $k$
times $t-1$ additions and $t$ multiplications, for each
$x \in \mathbb{F}_2^n$. But if we first compute $Ax^\top$, where $x$
has weight one, this only outputs the corresponding column of $A$ and
has no cost. From there we can compute the sums of two columns of $A$,
there are $\binom{n}{2}$ many of these sums and each one costs $k$
additions. From there we can compute all sums of three columns of $A$,
which are $\binom{n}{3}$ many and using the sums of two columns we
have already computed, means we only need to add one more column
costing $k$ additions. Proceeding in this way, until one reaches the
weight $t$, to compute $Ax^\top$ for all $x \in \mathbb{F}_2^n$, of
weight $t$ costs $k \cdot ( L(n,t) -n)$ additions, where
\begin{equation*}
  L(n,t) = \sum_{i=1}^t \binom{n}{i}.
\end{equation*}

This changes slightly over a general finite field. As a first step one
computes $Ax^\top$ for all $x \in \mathbb{F}_q^n$, of weight
$1$. Hence this step is no longer for free, but rather means computing
$A \lambda$ for all $\lambda \in \mathbb{F}_q^\star$, costing
$(q-1)kn$ multiplications. From there on one computes the sum of two
multiple of the columns, there are $\binom{n}{2}(q-1)^2$ many and each
sum costs $k$ additions. Hence proceeding in the same manner the cost
turns out to be $(q-1)kn$ multiplications and $(\bar{L}(n,t)-n(q-1))k$
additions,, where
\begin{equation*}
  \bar{L}(n,t) = \sum_{i=1}^t \binom{n}{i}(q-1)^i.
\end{equation*}

The next concept called early abort is also important whenever a
computation is done while checking the weight of the result. For
example one wants to compute $x+y$, where $x,y \in \mathbb{F}_2^n$,
which usually costs $n$ additions, but we only proceed in the
algorithm if $w(x+y) = t$. Hence we compute and check the weight
simultaneously and if the weight of the partial solution exceeds $t$
one does not need to continue. Over the binary one expects a randomly
chosen bit to have weight 1 with probability $\frac{1}{2}$, hence
after $2t$ we should reach the wanted weight $t$, and after $2(t+1)$
we should exceed the weight $t$. Hence on average we expect to compute
only $2(t+1)$ many bits of the solution, before we can abort. Over
$\mathbb{F}_q$, we expect a randomly chosen bit to have weight 1 with
probability $\frac{q-1}{q}$, therefore we need to compute
$\frac{q}{q-1}(t+1)$ many bits before we can abort.

An important step in the ball-collision algorithm is to check for a
collision, i.e. if $Ax^\top = By^\top$ one continues, where again
$A, B \in \mathbb{F}_2^{k \times n}$ and $x,y$ are living in some sets
$S$ and $T$ respectively. There are $\mid S \mid \cdot \mid T \mid$
many choices for $(x,y)$, assuming that they are distributed uniformly
over $\mathbb{F}_2^n$, then on average one expects the number of
collisions to be $\mid S \mid \cdot \mid T \mid 2^{-n}$.  Similarly
over $\mathbb{F}_q$ the number of collisions will be
$ \mid S \mid \cdot \mid T \mid q^{-n}.$

\section{Generalization of the Ball-Collision Algorithm}\label{sec3}

In this section we generalize the ball-collision algorithm over the
binary \cite{ballcoll} to a general finite field.

The algorithm requires a parity check matrix $H$. Notice that if the
generator matrix $G$ is published, the easiest way to get $H$ is to
choose an information set $I$ and to compute $\tilde{G}:=G_I^{-1}G$.

\begin{algorithm}[h!]
  \caption{Ball-collision over $\mathbb{F}_q$}\label{ballcoll2}
  \begin{flushleft}
    Input: The $(n-k)\times n$ parity check matrix $H$, the syndrome
    $s\in\mathbb{F}_q^{n-k}$ and the positive integers
    $p_1,\, p_2,\, q_1,\, q_2,\, k_1,\, k_2,\, \ell_1,\, \ell_2 \in
    \mathbb{Z}$,
    such that $k=k_1+k_2$, $p_i\leq k_i$, $q_i\leq \ell_i$ and
    $t-p_1-p_2-q_1-q_2\leq n-k-\ell_1-\ell_2$.

    Output: $e\in\mathbb{F}_q^n$ with $He^\top=s$ and $w(e)=t$.
  \end{flushleft}
  \begin{algorithmic}[1]
    \State Choose an information set $I\subseteq \{1, ..., n\}$ of $H$
    of size $k$.  \State Partition $I$ into two disjoint subsets $X_1$
    and $X_2$ of size $k_1$ and $k_2=k-k_1$ respectively.  \State
    Partition $Y=\{1, ..., n\}\backslash I$ into disjoint subsets
    $Y_1$ of size $\ell_1$, $Y_2$ of size $\ell_2$ and $Y_3$ of size
    $\ell_3=n-k-\ell_1-\ell_2$.  \State\label{find U} Find an
    invertible matrix $U\in\mathbb{F}_q^{(n-k)\times(n-k)} $, such
    that $(UH)_Y=\mathbf{1}_{n-k}$ and
    $(UH)_I=\begin{bmatrix} A_1 \\ A_2
    \end{bmatrix}$,
    where $A_1\in\mathbb{F}_q^{(\ell_1+\ell_2)\times k}$ and
    $A_2\in\mathbb{F}_q^{\ell_3\times k}$.  \State Compute
    $Us=\begin{bmatrix} s_1\\ s_2
    \end{bmatrix}$,
    where $s_1\in\mathbb{F}_q^{\ell_1+\ell_2}$ and
    $s_2\in\mathbb{F}_q^{\ell_3}$.  \For{each of \Statex-
      $V_1 \subset X_1$ of size $p_1$ \Statex- $V_2 \subset X_2$ of
      size $p_2$ \Statex- $W_1 \subset Y_1$ of size $q_1$ \Statex-
      $W_2 \subset Y_2$ of size $q_2$} \Statex Compute the following
    sets: \label{buildST} \Statex
    -$S=\{(A_1(\pi_I(x_1))+\pi_{Y_1 \cup Y_2}(y_1), x_1, y_1) \ | \
    x_1 \in\mathbb{F}_q^n(V_1), y_1 \in \mathbb{F}_q^n(W_1) \}$,
    \Statex
    -$T=\{(-A_1(\pi_I( x_2))+s_1-\pi_{Y_1 \cup Y_2}(y_2), x_2, y_2) \
    | \ x_2 \in\mathbb{F}_q^n(V_2), y_2\in \mathbb{F}_q^n(W_2)\}$.
    \EndFor
    \For{$(v, x_1, y_1)\in S$} \For{$(v, x_2, y_2)\in T$}
    \label{collT}
    \If{$w(-A_2(\pi_I( x_1+
      x_2))+s_2)=t-p_1-p_2-q_1-q_2$}\label{compute
      A2x2} \Statex \tab[1.5cm] Output:
    $e= x_1+x_2+ y_1+ y_2+ \sigma_{Y_3}(-A_2(\pi_I( x_1+ x_2))+s_2)$
    \Else \ go to Step 1 and choose new information set $I$.
    \EndIf
    \EndFor
    \EndFor
  \end{algorithmic}
\end{algorithm}

Again, as in the binary case, the idea of the algorithm is to solve
$UHe^\top=Us$ instead of $He^\top=s$, where an invertible $U$ is
chosen such that
$UH=\left[\begin{array}{cc} A &\mathbf{1}_{n-k} \end{array}\right]$
and $Us=\left[\begin{array}{c} s_1 \\ s_2
              \end{array}\right]$ 
            with
            $s_1 \in \mathbb{F}_q^{\ell_1+\ell_2},\,
            s_2\in\mathbb{F}_q^{\ell_3}$.
            We are therefore looking for a vector $e\in\mathbb{F}_q^n$
            fulfilling

$$UHe^\top=\left[\begin{array}{ccc}
                   A_1 &\mathbf{1}_{\ell_1+\ell_2} &0 \\
                   A_2 &0 &\mathbf{1}_{\ell_3}
                 \end{array}
               \right] \left[\begin{array}{c} e_1 \\ e_2 \\e_3
                             \end{array}\right]
                           =\left[\begin{array}{c}s_1 \\ s_2
                                  \end{array}\right],
$$
with
$e_1\in\mathbb{F}_q^k,\, e_2\in\mathbb{F}_q^{\ell_1+\ell_2},\,
e_3\in\mathbb{F}_q^{\ell_3}$.
This leads to the following system of equations:
\begin{align*}
  A_1e_1+e_2&=s_1, \\
  A_2e_1+e_3&=s_2. 
\end{align*}
The algorithm solves the above by finding
\begin{align*}
  &e_1=\pi_I(x_1+ x_2),  \\
  &e_2=\pi_{Y_1 \cup Y_2}( y_1 + y_2), \\
  &e_3=-A_2(\pi_I(x_1+x_2)) +s_2, 
\end{align*}
such that
\begin{align*}
  &A_1(\pi_I(x_1))+\pi_{Y_1 \cup Y_2}(y_1)=s_1-A_1(\pi_I(x_2))-\pi_{Y_1\cup Y_2}( y_2). 
\end{align*}
This last condition is fulfilled by the collision between $S$ and $T$ in Step \ref{collT}. \\
Observe that for $q=2$ the above algorithm is equivalent to the one
proposed over the binary. We hence did not change it in its
substantial form.
\\
We now want to prove that the ball-collision algorithm over
$\mathbb{F}_q$ works, i.e. that it returns any vector $e$ of the
desired form, if it exists. For this we follow the idea of
\cite{ballcoll}.

\begin{theorem}
  The ball-collision algorithm over $\mathbb{F}_q$ finds any vector
  $e$ that fulfills $UHe^\top=Us$ and is of the desired form - of
  weight $t$, with $p_1, p_2, q_1, q_2$ and $t-p_1-p_2-q_1-q_2$
  nonzero entries in $X_1, X_2, Y_1, Y_2$ and $Y_3$ respectively.
\end{theorem}
\begin{proof}
  First, we want to prove, that the output $e$ is of the desired form:\\
  \begin{itemize}
  \item $x_1$ is of weight $p_1$ and in $\mathbb{F}_q^n(X_1)$,
  \item $ x_2$ is of weight $p_2$ and in $\mathbb{F}_q^n(X_2)$,
  \item $y_1$ is of weight $q_1$ and in $\mathbb{F}_q^n(Y_1)$,
  \item $ y_2$ is of weight $q_2$ and in $\mathbb{F}_q^n(Y_2)$,
  \item
    $w(-A_2(\pi_I(x_1 +x_2)) +s_2)=w(A_2(\pi_I( x_1 + x_2))
    -s_2)=t-p_1-p_2-q_1-q_2$ and it lies in $\mathbb{F}_q^n(Y_3)$.
  \end{itemize}
  As the above subspaces do not intersect, $w(e)$ can be calculated by
  adding up the weights of each of them. Hence $w(e)=t$ and each of
  the subspaces has the desired weight distribution by definition.

  It remains to prove that $UHe^\top=Us$.  Let us write each of the
  subspaces $\mathbb{F}_q^n(I), \mathbb{F}_q^n(Y_1\cup Y_2)$ and
  $\mathbb{F}_q^n(Y_3)$ separately.
  \begin{align*}
    UHe^\top&=\begin{bmatrix}
      A_1 &\mathbf{1}_{\ell_1+\ell_2} &0 \\
      A_2 &0 &\mathbf{1}_{\ell_3}
    \end{bmatrix}\begin{bmatrix}
      \pi_I( x_1+ x_2)\\
      \pi_{Y_1 \cup Y_2}( y_1+y_2)\\
      -A_2(\pi_I(x_1 + x_2))+s_2
    \end{bmatrix}\\
            &=\begin{bmatrix}
              A_1(\pi_I( x_1+ x_2)) +\pi_{Y_1 \cup Y_2}( y_1+ y_2)\\
              A_2(\pi_I(x_1+ x_2))-A_2(\pi_I( x_1 + x_2))+s_2
            \end{bmatrix}\\
            &=\begin{bmatrix}
              A_1(\pi_I( x_1+ x_2)) +\pi_{Y_1 \cup Y_2}(y_1+y_2)\\
              s_2
            \end{bmatrix}.
  \end{align*}
  And we know that
  $A_1(\pi_I(x_1+ x_2)) +\pi_{Y_1 \cup Y_2}( y_1+ y_2)=s_1$ by the
  collision of $S$ and $T$ in Step \ref{collT}.

  We now want to prove that the algorithm returns each of the above
  vectors such that $He^\top=s$ under the assumption, that we worked
  with a correct partitioning into $X_1, X_2, Y_1, Y_2, Y_3$. We do
  that by checking whether the algorithm considers all possible
  combinations and does not exclude any possible solution.

  $U$ is invertible and hence does not exclude any solution when
  multiplied to $H$ and $s$. In Step \ref{buildST}, where we build the
  sets $S$ and $T$, we go over all the possible sets $V_1, V_2, W_1$
  and $W_2$, which contain all possible vectors of the desired weight
  distribution. There are only two steps in the algorithm, where we
  exclude certain vectors:
  \begin{enumerate}
  \item When we only keep the collisions between $S$ and $T$ in Step
    \ref{collT}. But this is justified as $A_1e_1+e_2 =s_1$, i.e.
    \begin{equation*}
      A_1(\pi_I( x_1))+\pi_{Y_1 \cup Y_2}(y_1)=-A_1(\pi_I(x_2))+s_1-\pi_{Y_1 \cup Y_2}(y_2)
    \end{equation*}
    needs to be satisfied.
  \item When we check whether
    $w(-A_2(\pi_I(x_1 + x_2))+s_2)=t-p_1-p_2-q_1-q_2$. But also this
    is justified as $e_3 \in \mathbb{F}_q^{\ell_3}$ needs this weight
    to complete the weight of $e$ to be $t$.
  \end{enumerate}
  Hence we consider all possible error vectors that are of the given
  weight distribution and satisfy $UHe^\top=Us$.
\end{proof}

\section{Complexity Analysis} \label{sec:complexity} In this section
we want to analyze the complexity of the extended ball-collision
algorithm over $\mathbb{F}_q$. Since the cost will be given in
operations over $\mathbb{F}_q$, we will denote by $\mathcal{M}$ the
multiplications needed and by $\mathcal{A}$ the amount of additions.
Note that one addition over $\mathbb{F}_q$ costs $\log_2(q) $ bit
operations and one multiplication over $\mathbb{F}_q$ costs
$\log_2(q)\log_2(\log_2(q))\log_2(\log_2(\log_2(q)))$ bit operations.

\subsubsection*{Success Probability of one Iteration}

We follow the idea of \cite{ballcoll} as the success probability does
not depend on the base field, in fact: we have the same success
probability over $\mathbb{F}_q$ as over $\mathbb{F}_2$, since it only
depends on choosing the correct partition of the subspaces. The
success probability of one iteration equals the chances that there are
$p_i$ error bits in $X_i$, $q_i$ error bits in $Y_i$ and the remaining
ones in $Y_3$ - all for $i\in\{1, 2\}$. If this distribution is
assumed correctly, then the algorithm will find the error vector $e$
as it goes over all possible combinations of vectors in each of the
mentioned subspaces.  Hence the iteration succeeds with a probability
of

\begin{equation}\label{probsuccess fq}
  \binom{n}{t}^{-1} \binom{\ell_3}{t-p_1-p_2-q_1-q_2}\binom{k_1}{p_1}
\binom{k_2}{p_2}\binom{\ell_1}{q_1}\binom{\ell_2}{q_2}.
\end{equation}

\subsubsection*{Cost of one Iteration}

In Step \ref{find U} of the algorithm, one uses Gaussian elimination
to find an invertible matrix $U$, bringing $H$ into systematic form,
since we will also need to compute $Us$ we will directly perform
Gaussian elimination on the matrix $\begin{pmatrix} H \mid & s
\end{pmatrix}$,
where we adjoined the vector $s$ as a column to $H$.  A crude estimate
of the cost for this step is
$(n-k)(n+1)(n-k-1)(\mathcal{A} + \mathcal{M})$.

To build the set $S$ we want to use the concept of intermediate sums
over $\mathbb{F}_q$ described before. Hence to compute
$A_1(\pi_I(x_1))$, for all $x_1 \in \mathbb{F}_q^n(V_1)$ we need
$(q-1)(\ell_1+\ell_2)k_1$ multiplications and
$(\bar{L}(k_1, p_1)-k_1(q-1))(\ell_1+\ell_2)$ additions. To a fixed
$A_1(\pi_I(x_1))$, we then add $\pi_{Y_1 \cup Y_2}(y_1)$ again using
intermediate sums this costs $\bar{L}(\ell_1, q_1)$ additions for each
of the $x_1 \in \mathbb{F}_q^n(V_1)$, which are
$\binom{k_1}{p_1}(q-1)^{p_1}$ many. Hence resulting in a total cost of
\begin{align*}
  & (q-1)(\ell_1+\ell_2)k_1\mathcal{M} +  
   (\bar{L}(k_1, p_1)-k_1(q-1))(\ell_1+\ell_2)\mathcal{A} \\
  & + \binom{k_1}{p_1}\bar{L}(\ell_1, q_1)(q-1)^{p_1} \mathcal{A}.
\end{align*}


To build the set $T$ we proceed similarly, the only difference being
that $s_1$ needs to be added to the first step of the intermediate
sums over $\mathbb{F}_q$, hence adding a cost of
$(\ell_1+\ell_2)(q-1)k_2$ additions. The total cost of this step is
hence given by
\begin{align*}
  & (q-1)(\ell_1+\ell_2)k_2(\mathcal{M} +  \mathcal{A}) 
   + (\bar{L}(k_2, p_2-k_2(q-1)))(\ell_1+\ell_2) \mathcal{A} \\
  & + \binom{k_2}{p_2}\bar{L}(\ell_2, q_2)(q-1)^{p_2} \mathcal{A}.
\end{align*}

In Step \ref{collT}, when checking for collisions between $S$ and $T$,
we want to calculate the number of collisions we can expect on
average. The elements in $S$ and $T$ are all of length $\ell_1+\ell_2$
and hence there is a total of $q^{\ell_1+\ell_2}$ possible
elements. $S$ has $\binom{k_1}{p_1}\binom{\ell_1}{q_1}(q-1)^{p_1+q_1}$
many elements and $T$ has
$\binom{k_2}{p_2}\binom{\ell_2}{q_2}(q-1)^{p_2+q_2}$ many elements, we
therefore get that the expected number of collisions is
\begin{align*}
  \frac{\binom{k_1}{p_1}\binom{k_2}{p_2}\binom{\ell_1}{q_1}
\binom{\ell_2}{q_2}(q-1)^{p_1+p_2+q_1+q_2}}{q^{\ell_1+\ell_2}}.
\end{align*}

For each collision we have, we check whether
$w(-A_2(\pi_I(x_1+x_2))+s_2)=t-p_1-p_2-q_1-q_2$ is satisfied. For this
we will use the method of early abort: to compute one bit of the
result costs $(p_1+p_2+1)$ additions and $(p_1+p_2)$ multiplications,
hence this step costs on average
\begin{equation*}
  \frac{q}{q-1}(t-p_1-p_2-q_1-q_2+1)\left((p_1+p_2+1)\mathcal{A} + 
(p_1+p_2)\mathcal{M}\right).
\end{equation*}

Hence the total cost of one iteration is given by

\begin{align}\label{costiteration}
  &(n-k)(n+1)(n-k-1)(\mathcal{A} + \mathcal{M}) \nonumber \\
  & +  (\ell_1+ \ell_2)[ (q-1)((k_1+k_2) \mathcal{M}+k_2\mathcal{A}) \nonumber \\
  & + (\bar{L}(k_1, p_1)-k_1(q-1))\mathcal{A} +  (\bar{L}(k_2, p_2)-k_2(q-1))\mathcal{A} ] \nonumber \\
  & + \binom{k_1}{p_1}\bar{L}(\ell_1, q_1)(q-1)^{p_1}  \mathcal{A}  + 
   \binom{k_2}{p_2}\bar{L}(\ell_2, q_2)(q-1)^{p_2} \mathcal{A} \\
  & +\binom{k_1}{p_1}\binom{k_2}{p_2}\binom{\ell_1}{q_1}
   \binom{\ell_2}{q_2}(q-1)^{p_1+p_2+q_1+q_2}q^{-(\ell_1+\ell_2)}  \nonumber \\
  & \cdot \frac{q}{q-1}(t-p_1-p_2-q_1-q_2+1)\left((p_1+p_2+1)\mathcal{A} + (p_1+p_2)\mathcal{M}\right). \nonumber
\end{align}

\subsubsection*{Overall cost}

Combining the result from \eqref{probsuccess fq} and
\eqref{costiteration} the overall cost of the ball-collision algorithm
over $\mathbb{F}_q$ then amounts to

\begin{align*}
  &\binom{n}{t}\left(\binom{\ell_3}{t-p_1-p_2-q_1-q_2}\binom{k_1}{p_1}
   \binom{k_2}{p_2}\binom{\ell_1}{q_1}\binom{\ell_2}{q_2}\right)^{-1} \\
  & \cdot [(n-k)(n+1)(n-k-1)(\mathcal{A} + \mathcal{M})  \\
  & +  (\ell_1+ \ell_2)[ (q-1)((k_1+k_2) \mathcal{M}+k_2\mathcal{A})  \\
  & + (\bar{L}(k_1, p_1)-k_1(q-1))\mathcal{A} +  (\bar{L}(k_2, p_2)-k_2(q-1))\mathcal{A} ]  \\
  & + \binom{k_1}{p_1}\bar{L}(\ell_1, q_1)(q-1)^{p_1}  \mathcal{A}  
   + \binom{k_2}{p_2}\bar{L}(\ell_2, q_2)(q-1)^{p_2} \mathcal{A} \\
  & +\binom{k_1}{p_1}\binom{k_2}{p_2}\binom{\ell_1}{q_1}
   \binom{\ell_2}{q_2}(q-1)^{p_1+p_2+q_1+q_2}q^{-(\ell_1+\ell_2)}   \\
  & \cdot \frac{q}{q-1}(t-p_1-p_2-q_1-q_2+1)\left((p_1+p_2+1)\mathcal{A} + (p_1+p_2)\mathcal{M}\right)].
\end{align*}

\subsection{Asymptotic Complexity}

In this subsection we want to find the asymptotic complexity of the
ball-collision algorithm over
$\mathbb{F}_q$. 

Fix real numbers $0 < T < 1/2$ and $R$,
with $$-T\log_q(T)-(1-T)\log_q(1-T) \leq 1-R<1.$$

We consider codes of large length $n$, we fix functions
$k,t: \mathbb{N} \to \mathbb{N}$ which satisfy
$\lim_{n \to \infty}t(n)/n =T$ and $\lim_{n \to \infty}k(n)/n = R$.

We fix real numbers $P,Q,L$ with $0 \leq P \leq R/2, 0 \leq Q \leq L$
and $$0 \leq T - 2P-2Q \leq 1 - R - 2L.$$ We fix the parameters
$p_1,p_2, q_1,q_2,\ell_1, \ell_2, k_1, k_2$ of the ball-collision
algorithm over $\mathbb{F}_q$ such that
\begin{itemize}
\item[i)] $\lim_{n \to \infty} \frac{p_i}{n} = P, $
\item[ii)] $\lim_{n \to \infty} \frac{q_i}{n} = Q, $
\item[iii)] $\lim_{n \to \infty} \frac{k_i}{n} = R/2, $
\item[iv)] $\lim_{n \to \infty} \frac{\ell_i}{n} = L,$
\end{itemize}
for $i \in \{1,2\}$. We use the convention that $x \log_q(x) = 0$, for
$x=0$.  In what follows we will use the following asymptotic formula
for binomial coefficients:
\begin{equation*}
  \lim_{n \to \infty} \frac{1}{n} \log_q \binom{ \alpha + o(1)n}{\beta + o(1)n} = 
   \alpha \log_q(\alpha)-\beta\log_q(\beta)-(\alpha-\beta)\log_q(\alpha-\beta).
\end{equation*}

With this formula we get the following:
\begin{itemize}
\item[i)]
  $\lim_{n \to \infty} \frac{1}{n} \log_q \binom{n}{t}= -T
  \log_q(T)-(1-T)\log_q(1-T),$
\item[ii)]
  $\lim_{n \to \infty} \frac{1}{n} \log_q \binom{k_i}{p_i}=
  R/2\log_q(R/2)- P \log_q(P)-(R/2-P)\log_q(R/2-P),$
\item[iii)]
  $\lim_{n \to \infty} \frac{1}{n} \log_q \binom{\ell_i}{q_i}=
  L\log_q(L)- Q \log_q(Q)-(L-Q)\log_q(L-Q),$
\item[iv)]
  $\lim_{n \to \infty} \frac{1}{n} \log_q
  \binom{n-k-\ell_1-\ell_2}{t-p_1-p_2-q_1-q_2}=
  (1-R-2L)\log_q(1-R-2L)- (T-2P-2Q)
  \log_q(T-2P-2Q)-(1-R-2L-T+2P+2Q)\log_q(1-R-2L-T+2P+2Q).$
\end{itemize}

\subsection*{Success probability} 
We will denote by $S(P,Q,L)$ the asymptotic exponent of the success
probability:
\begin{eqnarray*}
  S(P,Q,L) &=& \lim_{n \to \infty}\frac{1}{n} \log_q \left( \binom{n}{t}^{-1}\binom{n-k-\ell_1-\ell_2}{t-p_1-p_2-q_1-q_2} 
    \binom{k_1}{p_1}\binom{k_2}{p_2}\binom{\ell_1}{q_1} \binom{\ell_2}{q_2} \right) \\
           &=& T\log_q(T)+(1-T)\log_q(1-T) +(1-R-2L)\log_q(1-R-2L) \\  & & - (T-2P-2Q)\log_q(T-2P-2Q) \\
           & & - (1-R-2L-T+2P+2Q)\log_q(1-R-2L-T+2P+2Q) \\
           & & + R\log_q(R/2)-2P\log_q(P)-(R-2P)\log_q(R/2-P) +2L\log_q(L) \\
           & & -2Q\log_q(Q)-2(L-Q)\log_q(L-Q).
\end{eqnarray*}

\subsection*{Cost of one iteration} 
We will denote by $C(P,Q,L)$ the asymptotic exponent of the cost of
one iteration.
\begin{eqnarray*}
  C(P,Q,L) & = & \lim_{n \to \infty}\frac{1}{n} \log_q \left(  \binom{k_1}{p_1}(q-1)^{p_1} +
          \binom{k_2}{p_2}(q-1)^{p_2}+\binom{k_1}{p_1}\binom{\ell_1}{q_1}(q-1)^{p_1+q_1} \right.  \\
           & & \left. + \binom{k_2}{p_2} \binom{\ell_2}{q_2}(q-1)^{p_2+q_2} + \binom{k_1}{p_1}\binom{k_2}{p_2}\binom{\ell_1}{q_1}
              \binom{\ell_2}{q_2}(q-1)^{p_1+p_2+q_1+q_2}q^{-\ell_1-\ell_2} \right) \\
           &=& \max \left\{ \log_q(q-1)P + R/2\log_q(R/2) - P \log_q(P)-(R/2-P) \log_q(R/2-P), \right. \\
           & & \log_q(q-1)(P+Q)+R/2\log_q(R/2) - P \log_q(P)-(R/2-P) \log_q(R/2-P) \\
           & & + L \log_q(L) -Q \log_q(Q)-(L-Q)\log_q(L-Q), \\
           & & \log_q(q-1)(2P +2Q)-2L+R\log_q(R/2) -2P\log_q(P) \\
           & &  \left. -(R-2P)\log_q(R/2-P)+2L\log_q(L)-2Q\log_q(Q)-(2L-2Q)\log_q(L-Q)\right\}.
\end{eqnarray*}


\subsection*{Overall cost}

The overall asymptotic cost exponent of the ball-collision algorithm
over $\mathbb{F}_q$ is given by the difference of $C(P,Q,L)$ and
$S(P,Q,L)$:
\begin{eqnarray*}
  D(P,Q,L) &=& \max \left\{ \log_q(q-1)P- R/2\log_q(R/2)+P\log_q(P)+(R/2-P)\log_q(R/2-P) \right. \\
           & & -2L\log_q(L)+2Q\log_q(Q)+2(L-Q)\log_q(L-Q),  \\
           & & \log_q(q-1)(P + Q) -R/2\log_q(R/2)+ P \log_q(P) + (R/2-P) \log_q (R/2-P) \\
           & & -L\log_q(L) + Q\log_q(Q) +(L-Q)\log_q(L-Q), \\
           & &  \left. \log_q(q-1)(2P+2Q)-2L \right\} -T \log_q(T) - (1-T)\log_q(1-T) \\
           & & (1-R-2L)\log_q(1-R-2L) + (T - 2P-2Q) \log_q(T-2P-2Q) \\
           & & +(1-R-2L-T+2P+2Q)\log_q(1-R-2L-T+2P+2Q).
\end{eqnarray*}
The asymptotic complexity is then given by $q^{D(P,Q,L)n + o(n)}$.

Asymptotically, we assume that the code attains the Gilbert-Varshamov
bound, i.e. the code rate $R = k/n$ and the distance $D = d/n$ relate
via: \begin{equation} R = 1 + D\log_q(D)+(1-D) \log_q(1-D) - D
  \log_q(q-1). \label{eq:gv_bound}
\end{equation}  

In order to compute the asymptotic complexity of half-distance
decoding (i.e.  $T = D/2$) for a fixed rate $R$, we performed a
numerical optimization of the parameters $P,Q$ and $L$ such that the
overall cost $D(P,Q,L)$ is minimized subject to the following
constraints:
\[0 \leq P \leq R/2, 0 \leq Q \leq L \mbox{ and } 0 \leq T - 2P-2Q
\leq 1 - R - 2L.\]

Let $F(q,R)$ be the exponent of the optimized asymptotic
complexity. The asymptotic complexity of half-distance decoding at
rate $R$ over $\mathbb{F}_q$ is then given by $q^{F(q,R)n+o(n)}$.

In Table \ref{table_compare}, the values refer to the exponent of the
worst-case complexity of distinct algorithms, i.e. $F(q,R_{w})$ where
$R_{w} = {\rm argmax}_{0<R<1} \left (F(q,R)\right
)$. 
It compares Peter's generalization of Stern's algorithm to
$\mathbb{F}_q$, Hirose 's generalization of Stern's algorithm using
May-Ozerov's nearest neighbor algorithm (MO) to $\mathbb{F}_q$, Gueye
\textit{et al.} generalization of the algorithm of BJMM using MO to
$\mathbb{F}_q$ and the generalization of the ball-collision algorithm
to $\mathbb{F}_q$.

\begin{table}[H]
  \begin{tabular}{c|c|c | c | c}
    $q$ & $q-$Stern & $q-$Stern-MO & $q-$BJMM-MO & $q-$Ball-collision  \\ \hline
    \rule{0pt}{1.1\normalbaselineskip}
    2 &  $0.05563$ & $0.05498$ & $0.04730$ & $0.055573$ \\
    3 &  $0.05217$ & $0.05242$ & $0.04427$ & $0.052145$ \\
    4 &  $0.04987$ &  $0.05032$&  $0.04294$& $0.049846$ \\
    5 &  $0.04815$ & $0.04864$ &  $0.03955$ & $0.048140$ \\
    7 & $0.04571$ &  $0.04614$&  $0.03706$ & $0.045697$ \\
    8 &  $0.04478$ &  $0.04519$& $0.03593$ & $0.044770$ \\
    11 &  $0.04266$ &  $0.04299$&  $0.03335$ & $0.042656$ \\
    \hline
  \end{tabular}
  \caption{Comparison of the asymptotic complexities over 
$\mathbb{F}_q$. The values $q-$Stern and $q-$Stern-MO are from 
\cite[Table 1]{hirose}, and $q-$BJMM-MO are from \cite[Table 3]{klamti}.}\label{table_compare}
\end{table}



We can observe that the ball-collision algorithm over $\mathbb{F}_q$
outperforms Peter's generalization of Stern's algorithm to
$\mathbb{F}_q$ and Hirose's ISD algorithm over $\mathbb{F}_q$, for all
$q \geq 3$. Like in the binary case, the ball-collision algorithm does
not outperform the generalization of Gueye \textit{et al.} of the BJMM
algorithm using MO to $\mathbb{F}_q$.

\section*{Acknowledgments}
The fourth author is thankful to the Swiss National Science Foundation
grant number 169510.

\bibliographystyle{plain} \bibliography{biblio}

\end{document}